\documentclass[final,3p,times]{cd}

\usepackage{algorithm,algpseudocode}
\usepackage{graphicx}
  \graphicspath{{./pics/}}
  \DeclareGraphicsExtensions{.pdf}
\usepackage{amssymb}
\usepackage{amsmath}
\usepackage{multicol}
\usepackage[caption=false,font=footnotesize]{subfig}
\usepackage{stmaryrd}

\newcommand{\eps}{\varepsilon}
\newtheorem{theorem}{Theorem}
\newtheorem{lemma}[theorem]{Lemma}

\newtheorem{corollary}[theorem]{Corollary}

\newtheorem{definition}[theorem]{Definition}
\newtheorem{example}[theorem]{Example}

\newproof{proof}{Proof}

\journal{A revision with minor corrections, originally published in Systems \& Control Letters 61 (12), 1260-1268, 2012}

\begin{document}

\begin{frontmatter}

\title{On Conditional Decomposability}

\author[inst]{Jan~Komenda}
  \ead{komenda@ipm.cz}
  \address[inst]{Institute of Mathematics, Academy of Sciences of the Czech Republic, {\v Z}i{\v z}kova 22, 616 62 Brno, Czech Republic}
\author[inst]{Tom{\' a}{\v s}~Masopust\corref{cor1}}
  \ead{masopust@math.cas.cz}
\author[cwi]{Jan H. van Schuppen}
  \ead{J.H.van.Schuppen@cwi.nl}
  \address[cwi]{CWI, P.O. Box 94079, 1090 GB Amsterdam, The Netherlands}
\cortext[cor1]{Corresponding author. Institute of Mathematics, Academy of Sciences of the Czech Republic, {\v Z}i{\v z}kova 22, 616 62 Brno, Czech Republic, Tel.~+420222090784, Fax.~+420541218657 }

\begin{abstract}
  The requirement of a language to be conditionally decomposable is imposed on a specification language in the coordination supervisory control framework of discrete-event systems. In this paper, we present a polynomial-time algorithm for the verification whether a language is conditionally decomposable with respect to given alphabets. Moreover, we also present a polynomial-time algorithm to extend the common alphabet so that the language becomes conditionally decomposable. A relationship of conditional decomposability to nonblockingness of modular discrete-event systems is also discussed in this paper in the general settings. It is shown that conditional decomposability is a weaker condition than nonblockingness.
\end{abstract}
\begin{keyword}
  Discrete-event system \sep coordination control \sep conditional decomposability.
  \MSC 93C65 \sep 93A99 \sep 93B50
\end{keyword}
\end{frontmatter}

\section{Introduction}\label{intro}
  In the Ramadge-Wonham supervisory control framework, discrete-event systems are represented by deterministic finite automata. Given a specification language (usually also represented by a deterministic finite automaton), the aim of supervisory control is to construct a supervisor so that the closed-loop system satisfies the specification~\cite{RW87}. The theory is widely developed for the case where the system (plant) is monolithic. However, large engineering systems are typically constructed compositionally as a collection of many small components (subsystems) that are interconnected by rules; for instance, using a synchronous product or a communication protocol. This is especially true for discrete-event systems, where different local components run in parallel. Moreover, examples of supervisory control of modular discrete-event systems show that a coordinator is often necessary for achieving the required properties because the purely decentralized control architecture may fail in achieving these goals.
  
  The notion of separability of a specification language has been introduced in~\cite{WH91}, and says that a language $K$ over an alphabet $\bigcup_{i=1}^{n} E_i$, $n\ge 2$, is separable if $K=\bigparallel_{i=1}^{n} P_i(K)$, where for all $i=1,2,\ldots,n$, $P_i:(\bigcup E_i)^*\to E_i^*$ is a projection. A specification for a global system is separable if it can be represented (is fully determined) by local specifications for the component subsystems. It is very closely related to the notion of decomposability introduced in~\cite{RudieWonham90,RudieWonham1992} for decentralized discrete-event systems, which is also further studied in, e.g.,~\cite{JiangKumar2000}. Decomposability is a slightly more general condition because it involves not only the specification, but also the plant language, that is, a language $K\subseteq L$ over an alphabet $\bigcup_{i=1}^{n} E_i$, $n\ge 2$, is decomposable with respect to a plant language $L$ if $K=\bigparallel_{i=1}^{n} P_i(K) \parallel L$: separability is then decomposability where $L=(\bigcup_{i=1}^{n} E_i)^*$ is the set of all strings over the global alphabet. In this paper, we slightly abuse the terminology and call a separable language in the sense of~\cite{WH91} also decomposable. It has been shown in~\cite{WH91} that decomposability is important because it is computationally cheaper to compute locally synthesized supervisors that constitute a solution of the supervisory control problem for this decomposable specification. Recently, the notion of decomposability has also been extended to automata as an automaton decomposability in, e.g.,~\cite{karimadini:decomposability}. 

  However, the assumption that a specification language is decomposable is too restrictive. Therefore, several authors have tried to find alternative techniques for general indecomposable specification languages; for instance, the approach of~\cite{GM07} based on partial controllability, which requires that all shared events are controllable, or the shared events must have the same controllability status (but then an additional condition of so-called mutual controllability~\cite{KvSGM08} is needed).  

  In this paper, we study a weaker version of decomposability, so-called {\em conditional decomposability}, which has recently been introduced in~\cite{KvS08} and studied in~\cite{scl2011,Automatica2011} in the context of coordination supervisory control of discrete-event systems. It is defined as decomposability with respect to local alphabets augmented by the coordinator alphabet. The word conditional means that although a language is not decomposable with respect to the original local alphabets, it becomes decomposable with respect to the augmented ones, i.e., decomposability is only guaranteed (conditioned) by local event set extensions by coordinator events.
  
  In the coordination control approach of modular discrete-event systems, the plant is formed as a parallel composition of two or more subsystems, while the specification language is represented over the global alphabet. Therefore, the property of conditional decomposability is required in this approach to distribute parts of the specification to the corresponding components to solve the problem locally. More specifically, we need to ensure that there exists a corresponding part of the specification for the coordinator and for each subsystem composed with the coordinator. Thus, if the specification is conditionally decomposable, we can take this decomposition as the corresponding parts for the subsystems composed with a coordinator and solve the problem locally.
    
  Conditional decomposability depends on the alphabet of the coordinator, which can always be extended so that the specification is conditionally decomposable. In the worst (but unlikely) case all events must be put into the coordinator alphabet to make a language conditionally decomposable.  But in the case when the coordinator alphabet would be too large it is better to divide the local subsystems into groups that are only loosely coupled and introduce several coordinators on smaller alphabets. In this paper, a polynomial-time algorithm is provided for the verification whether a language is conditionally decomposable. We make an important observation that the algorithm is linear in the number of local alphabets, while algorithms for checking similar properties (such as decomposability and coobservability) suffer from the exponential-time complexity with respect to the number of local alphabets. This algorithm is then modified so that it extends the coordinator alphabet to make the specification language conditionally decomposable. Furthermore, we discuss a relationship of conditional decomposability to nonblockingness of a coordinated system, where a coordinated system is understood as a modular system composed of two or more subsystems and a coordinator.
  
  Finally, since one of the central notions of this paper is the notion of a (natural) projection, the reader is referred to~\cite{gjtm2011} for more information on the state complexity of projected regular languages.
  
  The rest of this paper is organized as follows. In Section~\ref{preliminaries}, basic definitions and concepts of automata theory and discrete-event systems are recalled. In Section~\ref{cdecomposability}, a polynomial-time algorithm for testing conditional decomposability for a general monolithic system is presented. In Section~\ref{sec:extension}, this algorithm is modified to extend the coordinator alphabet so that the specification becomes conditionally decomposable. In Section~\ref{cdandnonblocking}, the relation of nonblockingness of a coordinated system with conditional decomposability is discussed. The conclusion with hints for future developments is presented in Section~\ref{conclusion}.

\section{Preliminaries and definitions}\label{preliminaries}
	In this paper, we assume that the reader is familiar with the basic concepts of supervisory control theory~\cite{CL08} and automata theory~\cite{salomaa}. For an alphabet $E$, defined as a finite nonempty set, $E^*$ denotes the free monoid generated by $E$, where the unit of $E^*$, the empty string, is denoted by $\eps$. A {\em language\/} over $E$ is a subset of $E^*$. A prefix closure $\overline{L}$ of a language $L\subseteq E^*$ is the set of all prefixes of all words of $L$, i.e., it is defined as the set $\overline{L}=\{ w\in E^* \mid \exists u\in E^* : wu\in L\}$. A language $L$ is said to be prefix-closed if $L=\overline{L}$.

  In this paper, the notion of a generator is used to denote an incomplete deterministic finite automaton. A {\em generator\/} is a quintuple $G=(Q,E,\delta,q_0,F)$, where $Q$ is a finite set of {\em states}, $E$ is an {\em input alphabet}, $\delta: Q \times E \to Q$ is a {\em partial transition function}, $q_0 \in Q$ is the {\em initial state}, and $F\subseteq Q$ is the set of {\em final or marked states}. In the usual way, $\delta$ is inductively extended to a function from $Q \times E^*$ to $Q$. The language {\em generated\/} by $G$ is defined as the set $L(G) = \{w\in E^* \mid \delta(q_0,w)\in Q\}$, and the language {\em marked\/} by $G$ is defined as the set $L_m(G) = \{w\in E^* \mid \delta(q_0,w)\in F\}$. Moreover, we use the predicate $\delta(q,a)!$ to denote that the transition $\delta(q,a)$ is defined in state $q\in Q$ for event $a\in E$.
	
	For a generator $G$, let ${\tt trim}(G)$ denote the {\em trim\/} of $G$, that is, a generator ${\tt trim}(G)$ such that $\overline{L_m({\tt trim}(G))}=L({\tt trim}(G))=\overline{L_m(G)}$. In other words, all reachable states of $G$ from which no marked state is reachable are removed (including the corresponding transitions), and only reachable states are considered in ${\tt trim}(G)$, see~\cite{CL08,Won04}. A generator $G$ is said to be {\em nonblocking\/} if $\overline{L_m(G)}=L(G)$. Thus, ${\tt trim}(G)$ is always nonblocking.
	
  A {\em (natural) projection\/} $P: E^* \to E_0^*$, where $E_0\subseteq E$ are alphabets, is a homomorphism defined so that $P(a)=\eps$, for $a\in E\setminus E_0$, and $P(a)=a$, for $a\in E_0$. The {\em inverse image\/} of the projection $P$, denoted by $P^{-1}:E_0^*\to 2^{E^*}$, is defined so that for a language $L$ over the alphabet $E_0$, the set $P^{-1}(L)=\{s\in E^* \mid P(s)\in L\}$. In what follows, we use the notation $P_j^i$ to denote the projection from $E_i$ to $E_j$, that is, $P^{i}_{j} : E_i^* \to E_j^*$. In addition, we use the notation $E_{i+j}=E_i\cup E_j$, and, thus, $P^{i+j}_{k}$ denotes the projection from $E_{i+j}$ to $E_k$. If the projection is from the union of all the alphabets, then we simply use the notation $P_i:(\bigcup_j E_j)^* \to E_i^*$.

	Let $L_1\subseteq E_1^*$ and $L_2\subseteq E_2^*$ be two languages. The {\em parallel composition of $L_1$ and $L_2$\/} is defined as the language \[L_1\parallel L_2 = P_1^{-1}(L_1) \cap P_2^{-1}(L_2)\,,\] where $P_1: (E_1\cup E_2)^*\to E_1^*$ and $P_2: (E_1\cup E_2)^*\to E_2^*$.
	A similar definition in terms of generators follows. Let $G_1=(X_1,E_1,\delta_1,x_{01},F_1)$ and $G_2=(X_2,E_2,\delta_2,x_{02},F_2)$ be two generators. The {\em parallel composition of $G_1$ and $G_2$\/} is the generator $G_1\parallel G_2$ defined as the accessible part of the generator $(X_1\times X_2, E_1\cup E_2, \delta, (x_{01}, x_{02}), F_1\times F_2)$, where
  \[
    \delta((x,y),e)=
    \left\{\begin{array}{ll}
      (\delta_1(x,e),\delta_2(y,e)), & \text{ if } \delta_1(x,e)! \text{ and } \delta_2(y,e)!;\\
      (\delta_1(x,e), y),            & \text{ if } \delta_1(x,e)! \text{ and } e\notin E_2;\\
      (x, \delta_2(y,e)),            & \text{ if } e\notin E_1    \text{ and } \delta_2(y,e)!;\\
      \text{undefined,}              & \text{ otherwise.}
    \end{array}\right.
  \]
  The automata definition is related to the language definition by the following properties: $L(G_1\parallel G_2)=L(G_1)\parallel L(G_2)$ and $L_m(G_1\parallel G_2)=L_m(G_1)\parallel L_m(G_2)$, see~\cite{CL08}. 
  
  The automata-theoretic concept of nonblockingness of a composition of two generators $G_1$ and $G_2$ is equivalent to the language-theoretic concept of nonconflictness of two languages $L_m(G_1)$ and $L_m(G_2)$ if the generators $G_1$ and $G_2$ are nonblocking. Recall that two languages $L_1$ and $L_2$ are {\em nonconflicting\/} if $\overline{L_1}\parallel \overline{L_2}=\overline{L_1\parallel L_2}$, cf.~\cite{Won04,FLT,FW2006}.
  
  Let $G$ be a generator and $P$ be a projection, then $P(G)$ denotes the minimal generator such that $L_m(P(G))=P(L_m(G))$ and $L(P(G))=P(L(G))$. For a construction of $P(G)$, the reader is referred to~\cite{CL08,Won04}. 

	Now, the main concept of interest of this paper, the concept of conditional decomposability, is defined. See also \cite{KvS08,scl2011,Automatica2011,JKTMJHvS_wodes2010} for the applications and further discussion concerning this concept.
	
	\begin{definition}[Conditional decomposability]\label{cd}
    A language $K$ over an alphabet $E_1\cup E_2\cup\ldots\cup E_n$, $n\ge 2$, is said to be {\em conditionally decomposable with respect $E_1$, $E_2$,\ldots, $E_n$, and $E_k$\/}, where $\bigcup_{i,j\in\{1,2,\ldots,n\}}^{i\neq j} (E_i\cap E_j) \subseteq E_k\subseteq \bigcup_{j=1}^{n} E_j$, if
    \[
      K = P_{1+k} (K) \parallel P_{2+k} (K) \parallel \ldots \parallel P_{n+k}(K)\,.
    \]
    Recall that $P_{i+k}$ denotes the projection from $\bigcup_{j=1}^{n} E_j$ to $E_{i+k}$.
  \end{definition}

  Note that $\parallel_{i=1}^{n} P_{i+k} (K) = (\parallel_{i=1}^{n} P_{i+k} (K)) \parallel P_k(K)$ because $P_{i+k}(K)\subseteq (P^{i+k}_k)^{-1}P_k(K)$, which follows from the fact that $P^{i+k}_k P_{i+k}(K)=P_k(K)$. Hence, $\parallel_{i=1}^{n} P_{i+k} (K)\subseteq P_k^{-1}P_k(K)$. Moreover, if the language $K$ is given as a parallel composition of $n$ languages (over the required alphabets), then it is conditionally decomposable.
  \begin{lemma}
		A language $K\subseteq (E_1\cup E_2\cup\ldots\cup E_n)^*$ is conditionally decomposable with respect to alphabets $E_1$, $E_2$,\ldots, $E_n$, $E_k$ if and only if there exist languages $M_{i+k}\subseteq E_{i+k}^*$, $i=1,2,\ldots,n$, such that $K=\parallel_{i=1}^{n} M_{i+k}$.
	\end{lemma}
  \begin{proof}
    If $K=\ \parallel_{i=1}^{n} P_{2+k}(K)$, define $M_{i+k}=P_{i+k}(K)$, for $i=1,2,\ldots,n$. On the other hand, assume that there exist languages $M_{i+k} \subseteq E_{i+k}^*$, $i=1,2,\ldots,n$, such that $K =\ \parallel_{i=1}^{n} M_{i+k}$. Obviously, $P_{i+k}(K)\subseteq M_{i+k}$, $i=1,2,\ldots,n$, which implies that $\parallel_{i=1}^{n} P_{i+k}(K)\subseteq K$. As it always holds that $K\subseteq P_{i+k}^{-1}[P_{i+k}(K)]$, the definition of the synchronous product implies that $K\subseteq\ \parallel_{i=1}^{n} P_{i+k}(K)$.
  \qed\end{proof}

	Note that $K=\ \parallel_{i=1}^{n} M_{i+k}$ implies that the languages $P_{i+k}(K)\subseteq M_{i+k}$, for $i=1,2,\ldots,n$, which means that $P_{i+k}(K)$ are the smallest languages whose parallel composition results in $K$. In other words, if $K$ is conditionally decomposable, then $P_{i+k}(K)$, $i=1,2,\ldots,n$, is the smallest decomposition of $K$ with respect to the corresponding alphabets.

\section{Polynomial Test of Conditional Decomposability}\label{cdecomposability}
  In this section, we first construct a polynomial-time algorithm for the verification of conditional decomposability for alphabets $E_1$, $E_2$, and $E_k$, that is, for the case $n=2$, and then show how this is used to verify conditional decomposability for a general $n\ge 2$. To this end, consider a language $L$ over $E_1\cup E_2$, marked by a generator $G$. To verify whether or not $L$ is conditionally decomposable with respect $E_1$, $E_2$, and $E_k$, we construct a new structure as a parallel composition of two copies of $G$, denoted $f_{i+k}(G)$, for $i=1,2$, (see Example~\ref{ex2} and Figure~\ref{fig03}) that simultaneously verifies that each word of $P_{1+k}(L)\parallel P_{2+k}(L)$ also belongs to $L=L_m(G)$; $f_{i+k}(G)$ is constructed from the generator $G$ by renaming each event $e\in E_{j-k}=E_j\setminus E_k$, $j\neq i$, by a new event $\tilde{e}\in\tilde{E}_{j-k}$. In other words, each event $e$ which is not observed by $G$ according to the observable alphabet $E_i\cup E_k$ is replaced with a new event. Thus, the copy $f_{i+k}(G)$ is over the alphabet $E_{i+k}\cup\tilde{E}_{j-k}$, as demonstrated in the following example.
  \begin{example}\label{ex2}
    Consider the language $L_m(G)$ marked by the generator $G$ depicted in Figure~\ref{fig02}\subref{GG}, where the corresponding alphabets are $E_1=\{a,b,d\}$, $E_2=\{a,c,d\}$, and $E_k=\{a,d\}$. The isomorphic generators $f_{1+k}(G)$ with renamed event $c$, and $f_{2+k}(G)$ with renamed event $b$ are depicted in Figure~\ref{fig02}\subref{GG1} and Figure~\ref{fig02}\subref{GG2}, respectively.
    \begin{figure}[ht]
      \centering
      \subfloat[Generator $G$.]{\label{GG}\includegraphics[scale=.9]{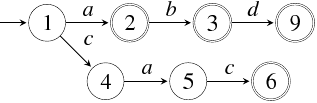}}
      \qquad
      \subfloat[Generator $f_{1+k}(G)$.]{\label{GG1}\includegraphics[scale=.9]{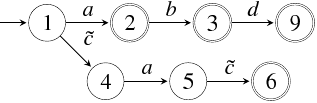}}
      \qquad
      \subfloat[Generator $f_{2+k}(G)$.]{\label{GG2}\includegraphics[scale=.9]{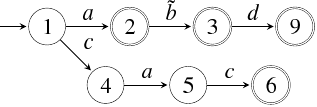}}
      \caption{Generators $G$, $f_{1+k}(G)$ and $f_{2+k}(G)$.}
      \label{fig02}
    \end{figure}
  $\hfill\diamond$\end{example}
  
  More specifically, let $E_1$, $E_2$, $E_k$ be alphabets such that $E_1\cap E_2 \subseteq E_k\subseteq E_1\cup E_2$, and define the global alphabet $E=E_1\cup E_2$. The structure is constructed as follows:
  \begin{enumerate}
    \item For the alphabet $E_i\setminus E_k$, where $i=1,2$, introduce a new alphabet $\tilde{E}_{i-k}=\{\tilde{a} \mid a\in E_i\setminus E_k\}$ that for each event $a\in E_i\setminus E_k$ contains a new event $\tilde{a}$. That is, $\tilde{E}_{i-k}\cap (E_i\setminus E_k)=\emptyset$ and there exists a bijection $g_{i-k}$ from $(E_i\setminus E_k)$ to $\tilde{E}_{i-k}$ such that $g_{i-k}(a)=\tilde{a}$. Note that $\tilde{E}_{1-k}\cap\tilde{E}_{2-k}=\emptyset$ because $E_1\cap E_2\subseteq E_k$. 
    
    \item Recall that $E_{i+k}=E_i\cup E_k$, for $i=1,2$, and let $\tilde{P}:(E\cup\tilde{E}_{1-k}\cup\tilde{E}_{2-k})^*\to E^*$ be a projection.
    
    \item Define two isomorphisms $f_{i+k}: E^* \to (E_{i+k}\cup \tilde{E}_{j-k})^*$, where $i,j\in\{1,2\}$, $i\neq j$, so that
    \[
      f_{i+k}(a)=
      \left\{\begin{array}{ll}
        a\,,         & \text{ for } a\in E_{i+k};\\
        \tilde{a}\,, & \text{ for } a\in E_j\setminus E_k\,.
      \end{array}\right.
    \]
    Note that it immediately follows that $\tilde{P}(f_{i+k}(L_m(G)))=P_{i+k}(L_m(G)))$ because both projections remove all events that are not in $E_{i+k}$.

    \item For a generator $G=(Q,E,\delta,q_0,F)$, we abuse the notation and denote by $f_{i+k}(G)=(Q,E_{i+k}\cup \tilde{E}_{j-k},\tilde{\delta},q_0,F)$, where $j\neq i$, the generator isomorphic with $G$ where events are renamed according to the isomorphism $f_{i+k}$, and the transition function $\tilde{\delta}$ is define as $\tilde{\delta}(q,f_{i+k}(a))=\delta(q,a)$.

    \item Let $L\subseteq E^*$ be a language generated by a minimal generator $G$, and define the generator \[\tilde{G} = f_{1+k}(G)\parallel f_{2+k}(G)\] over the alphabet $E\cup\tilde{E}_{1-k}\cup\tilde{E}_{2-k}$. By the definition of $\tilde{G}$, the assumption that $E_1\cap E_2\subseteq E_k$ which ensures that $\tilde{P}$ distributes over the synchronous product (see Lemma~\ref{lemma:Wonham} below), and Step 3 above, respectively, we have that
    \begin{equation}\label{eq1}
      \begin{split}
          \tilde{P}(L_m(\tilde{G}))
          &= \tilde{P}(f_{1+k}(L_m(G)) \parallel f_{2+k}(L_m(G)))\\
          &= \tilde{P}(f_{1+k}(L_m(G))) \parallel \tilde{P}(f_{2+k}(L_m(G))) \\
          &= P_{1+k}(L_m(G)) \parallel P_{2+k}(L_m(G))\,.
      \end{split}
    \end{equation}
  \end{enumerate}

  \begin{lemma}[\cite{Won04}]\label{lemma:Wonham}
    Let $E_1\cap E_2\subseteq E_k\subseteq E_1\cup E_2$, and let $L_1\subseteq E_1^*$ and $L_2\subseteq E_2^*$ be languages. Let $P_k : E^*\to E_k^*$ be a projection, then $P_k(L_1\parallel L_2)=P_{k}(L_1) \parallel P_{k}(L_2)$.
  \end{lemma}

  From the equations of $(\ref{eq1})$, we immediately have the following result for conditional decomposability.
  \begin{theorem}\label{lem1}
    The language $L_m(G)$ is conditionally decomposable with respect to alphabets $E_1$, $E_2$, $E_k$ if and only if it holds that $\tilde{P}(L_m(\tilde{G})) = L_m(G)$.
  \end{theorem}
  \begin{proof}
    The proof follows immediately from the definition of conditional decomposability and~(\ref{eq1}).
  \qed\end{proof}

  However, the inclusion $L_m(G)\subseteq P_{1+k}(L_m(G))) \parallel P_{2+k}(L_m(G))) = \tilde{P}(L_m(\tilde{G}))$ always holds. Thus, only the opposite inclusion is of interest. This inclusion, $\tilde{P}(L_m(\tilde{G}))\subseteq L_m(G)$, holds if and only if $L_m(\tilde{G})\subseteq \tilde{P}^{-1}(L_m(G))$, which gives the following key theorem for testing conditional decomposability.

  \begin{theorem}\label{thm3}
    The language $L_m(G)$ is conditionally decomposable with respect to alphabets $E_1$, $E_2$, $E_k$ if and only if the inclusion $L_m(\tilde{G})\subseteq \tilde{P}^{-1}(L_m(G))$ holds.
  \end{theorem}
  \begin{proof}
    It remains to prove that $\tilde{P}(L_m(\tilde{G}))\subseteq L_m(G)$ if and only if $L_m(\tilde{G})\subseteq \tilde{P}^{-1}(L_m(G))$. However, if $\tilde{P}(L_m(\tilde{G}))\subseteq L_m(G)$, then $L_m(\tilde{G})\subseteq \tilde{P}^{-1}\tilde{P}(L_m(\tilde{G}))\subseteq \tilde{P}^{-1}(L_m(G))$. On the other hand, assume that $L_m(\tilde{G})\subseteq \tilde{P}^{-1}(L_m(G))$. Then, $\tilde{P}(L_m(\tilde{G}))\subseteq \tilde{P}\tilde{P}^{-1}(L_m(G))=L_m(G)$.
  \qed\end{proof}

  The verification of this inclusion results in Algorithm~\ref{IsCD} for checking conditional decomposability of two components in polynomial time. Let a language $L$ be represented by the minimal generator $G=(Q,E,\delta,q_0,F)$ with the complete (total) transition function $\delta$ such that $L_m(G)=L$. If the transition function is not complete, the generator can be completed in time $O(|E|\cdot |Q|)$ by adding no more than one non-marked state and the missing transitions. Assume that the alphabets $E_1$, $E_2$, and $E_k$ are such that $E_1\cap E_2 \subseteq E_k\subseteq E_1\cup E_2=E$, and see Algorithm~\ref{IsCD}.
  \begin{algorithm}[ht]
  \caption{Conditional decomposability checking.}\label{IsCD}
    \begin{algorithmic}[1]
    \Function{IsCD}{$G,E_1,E_2,E_k$}
      \State Compute $\tilde{G}$                                                \Comment{$O(|E|\cdot |Q|^2)$.}
      \State Compute $\tilde{P}^{-1}(L_m(G))$                                   \Comment{$O(|E|\cdot |Q|)$.}
      \State Compute complement $co(\tilde{P}^{-1}(L_m(G)))$                    \Comment{$O(|Q|)$.}
      \If {$co(\tilde{P}^{-1}(L_m(G)))\cap L_m(\tilde{G}) = \emptyset$}         \Comment{$O(|E|\cdot |Q|^3)$.}
        \State \Return $L_m(G)$ is CD.
      \Else
        \State \Return $L_m(G)$ is not CD.
      \EndIf
    \EndFunction
    \end{algorithmic}
  \end{algorithm}
  To determine the time complexity of the algorithm, note that the computation is dominated by step~5, and thus the overall time complexity can be stated as $O(|E|\cdot |Q|^3)$. This also means that the space complexity is polynomial with respect to the number of states of the input generator $G$ because we do not need to use more space than $O(|E|\cdot |Q|^3)$.  The complexity of individual steps of the algorithm are computed as follows. Step~2 is a parallel composition of two copies of $G$, which requires to create up to $|Q|^2$ states of the generator $\tilde{G}$, and for each of these states up to $|E|$ transitions. Step~3 requires up to $|E|\cdot |Q|$ steps because in each state, we have to add self-loops labeled by the new symbols from $\tilde{E}_{1-k}\cup\tilde{E}_{2-k}$. The complement in Step~4 is computed by interchanging the marking of states, cf.~\cite{si97}. That is, marked states are unmarked and vice versa. As $G$ is complete, this results in a generator for the complement. Note that Steps~3 and~4 can be done at the same time. Finally, to decide the emptiness in Step~5 requires up to $|Q|^2 \cdot |Q|$ using a standard product automaton, see~\cite{si97}, where for each state, up to $|E|$ transitions are constructed, and is verified by the reachability of a final state by the depth-first-search procedure in linear time \cite{Cormen:2009:IAT:1614191}. Note also that it is a longstanding open problem whether the emptiness of intersection of two regular languages generated by generators with $m_1$ and $m_2$ states, respectively, can be decided in time $o(m_1\cdot m_2)$, cf.~\cite{lipton}. If this is possible, then the complexity of our algorithm can be improved accordingly.

  We demonstrate our approach in the following example.
  \begin{example}\label{ex2_2}
    Consider the language $L_m(G)$ marked by the generator $G$ depicted in Figure~\ref{fig02}\subref{GG}, where the corresponding alphabets are $E_1=\{a,b,d\}$, $E_2=\{a,c,d\}$, and $E_k=\{a,d\}$. The isomorphic generators $f_{1+k}(G)$ with renamed event $c$, and $f_{2+k}(G)$ with renamed event $b$ are depicted in Figure~\ref{fig02}\subref{GG1} and Figure~\ref{fig02}\subref{GG2}, respectively. Their parallel composition $\tilde{G}$ is shown in Figure~\ref{fig03}. It is obvious that the string ``$cacb$'' belongs to the language $L_m(\tilde{G})$, whereas it does not belong to the language $\tilde{P}^{-1}(L_m(G))$. Thus, by Theorem~\ref{thm3}, the language $L_m(G)$ is not conditionally decomposable with respect to alphabets $E_1$, $E_2$, $E_k$.
    \begin{figure}[ht]
      \centering
      \includegraphics[scale=.9]{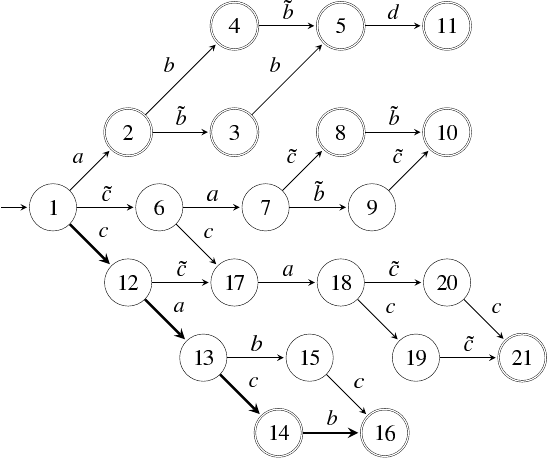}
      \caption{Generator $\tilde{G}=f_{1+k}(G)\parallel f_{2+k}(G)$ with a highlighted word violating conditional decomposability of the language $L_m(G)$.}
      \label{fig03}
    \end{figure}
  $\hfill\diamond$\end{example}

  Now, we generalize this approach to verifying conditional decomposability for a general number of $n\ge 2$ alphabets. The following theorem proves that we can directly use Algorithm~\ref{IsCD}.
  \begin{theorem}\label{thmgen}
    Let $K$ be a language, and let $E_i$, for $i=1,2,\ldots,n$, $n\ge 2$, and $E_k$ be alphabets such that $\bigcup_{i\neq j} (E_i\cap E_j) \subseteq E_k\subseteq \bigcup_{j=1}^{n} E_j$. Then, $P_{i+k}(K) \parallel P_{1+2+\ldots+(i-1)+(i+1)+\ldots+n+k}(K) \subseteq K$, for all $i=1,2,\ldots,n$, if and only if $K$ is conditionally decomposable with respect to alphabets $E_i$, $i=1,2,\ldots,n$, and $E_k$.
  \end{theorem}
  \begin{proof}
    First, $P_{1+2+\ldots+(i-1)+(i+1)+\ldots+n+k}(K) \subseteq P_{1+k}(K) \parallel P_{2+k}(K)\parallel\ldots\parallel P_{(i-1)+k}(K)\parallel P_{(i+1)+k}(K)\parallel \ldots\parallel P_{n+k}(K)$ because for all $j\in\{1,2,\ldots,n\}\setminus\{i\}$, we have $P_{j+k}(P_{1+2+\ldots+(i-1)+(i+1)+\ldots+n+k}(K)) = P_{j+k}(K)$. Thus, if $K$ is conditionally decomposable, then $P_{i+k}(K)\parallel P_{1+2+\ldots+(i-1)+(i+1)+\ldots+n+k}(K) \subseteq P_{i+k}(K)\parallel P_{1+k}(K) \parallel\ldots\parallel P_{(i-1)+k}(K)\parallel P_{(i+1)+k}(K)\parallel \ldots\parallel P_{n+k}(K) = K$, for all $i=1,2,\ldots,n$.
    
    To prove the opposite implication, assume that $K$ is not conditionally decomposable. Then there exist $t_i=P_{i+k}(w_i)$, for some $w_i\in K$ and for all $i=1,2,\ldots,n$, such that $t_1\parallel t_2\parallel\ldots\parallel t_n \not\subseteq K$. We prove by induction on $i=1,2,\ldots,n-1$ that
    \begin{align}\label{eqA}
      \{t_i\} \parallel \{t_{i-1}\}\parallel \ldots \parallel \{t_2\} \parallel \{t_1\}\parallel P_{(i+1)+(i+2)+\ldots+n+k}(w_n) \subseteq K\,.
    \end{align}
    For $i=1$ and by the assumption, $\{t_1\} \parallel P_{2+3+\ldots+n+k}(w_n) \subseteq P_{1+k}(K) \parallel P_{2+3+\ldots+n+k}(K) \subseteq K$. Thus, we assume that it holds for all $i=1,2,\ldots,\ell$, $\ell<n-1$, and we prove it for $i=\ell+1$. By the induction hypothesis, $\{t_\ell\} \parallel \{t_{\ell-1}\}\parallel \ldots \parallel \{t_2\} \parallel \{t_1\}\parallel P_{(\ell+1)+(\ell+2)+\ldots+n+k}(w_n) \subseteq K$. Then, using the projection $P_{1+2+\ldots+\ell+(\ell+2)+\ldots+n+k}$, we get that
    \begin{align*}
      P_{1+2+\ldots+\ell+(\ell+2)+\ldots+n+k}\Bigl(\{t_\ell\} \parallel \{t_{\ell-1}\}\parallel \ldots \parallel \{t_2\} \parallel \{t_1\}\parallel P_{(\ell+1)+(\ell+2)+\ldots+n+k}(w_n)\Bigr) \subseteq P_{1+2+\ldots+\ell+(\ell+2)+\ldots+n+k}(K)
    \end{align*}
    and, by Lemma~\ref{lemma:Wonham}, we get that $P_{1+2+\ldots+\ell+(\ell+2)+\ldots+n+k}\Bigl(\{t_\ell\} \parallel \{t_{\ell-1}\}\parallel \ldots \parallel \{t_2\} \parallel \{t_1\}\parallel P_{(\ell+1)+(\ell+2)+\ldots+n+k}(w_n)\Bigr) = \{t_\ell\} \parallel \{t_{\ell-1}\}\parallel \ldots \parallel \{t_2\} \parallel \{t_1\}\parallel P_{(\ell+2)+\ldots+n+k}(w_n)$. By this equality and the assumption for $i=\ell+1$, we have
    \begin{align*}
      & \{t_{\ell+1}\} \parallel \Bigl[\{t_\ell\} \parallel \{t_{\ell-1}\}\parallel \ldots \parallel \{t_2\} \parallel \{t_1\}\parallel P_{(\ell+2)+\ldots+n+k}(w_n)\Bigr]\\
      & \subseteq P_{(\ell+1)+k}(K) \parallel P_{1+2+\ldots+\ell+(\ell+2)+\ldots+n+k}(K)\\
      & \subseteq K
    \end{align*}
    as claimed. Then, substituting $i=n-1$ to $(\ref{eqA})$, we immediately have that $\{t_{n-1}\} \parallel \{t_{n-2}\}\parallel \ldots \parallel \{t_2\} \parallel \{t_1\}\parallel P_{n+k}(w_n) \subseteq K$, which together with $P_{n+k}(w_n)=t_n$ implies that $\{t_{n-1}\} \parallel \{t_{n-2}\}\parallel \ldots \parallel \{t_2\} \parallel \{t_1\}\parallel \{t_n\} \subseteq K$, which is a contradiction. Thus, $K$ is conditionally decomposable.
  \qed\end{proof}

  The previous theorem says that we can check conditional decomposability of a language $K$ by $n$ executions of Algorithm~\ref{IsCD}. This means that the overall complexity of verifying conditional decomposability for a general number of alphabets, $n\ge 2$, is $O(n\cdot |E|\cdot |Q|^3)$, which is polynomial with respect to the number of states and the number of components.
  
  To conclude this section, note that an example of an $r$-state automaton with $|E|=4$ and a projection reaching the exponential upper bound on the number of states, more precisely the upper bound $3\cdot 2^{r-2} - 1$, has been shown in \cite{wong98}. Thus, the approach following the definition of conditional decomposability computing projections and parallel composition is exponential for that language even for the case of two alphabets. In comparison, the complexity of our algorithm is polynomial. A preliminary version of this algorithm has been implemented in libFAUDES \cite{faudes}.

\section{Extension of the coordinator alphabet}\label{sec:extension}
  According to Theorem~\ref{thmgen}, we can again consider only the case $n=2$. To compute an extension of $E_k$ so that the language becomes conditionally decomposable, we modify Algorithm~\ref{IsCD} to Algorithm~\ref{extension}, which uses more structural properties of the structure $\tilde{G}$. First, however, we explain the technique on an example.
  \begin{figure}[hbt]
    \centering
    \includegraphics[scale=.9]{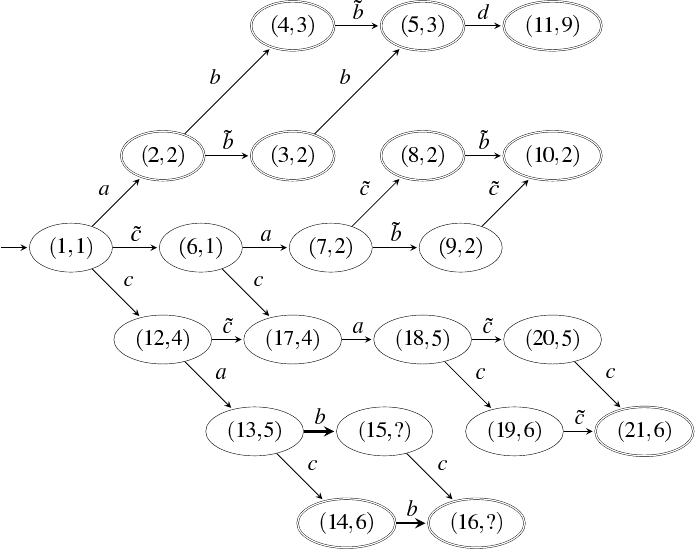}
    \caption{Generator $\tilde{G}$ with the corresponding states of $G\times \tilde{G}$. Note that transitions $\delta(5,b)$ and $\delta(6,b)$ are not defined in $G$, and, therefore, they violate conditional decomposability of the language $L_m(G)$.}
    \label{figextension}
  \end{figure}
  \begin{example}
    Consider the generator $G$ and $\tilde{G}$ of Examples~\ref{ex2} and~\ref{ex2_2}. The main idea of this technique is to construct, step-by-step, the parallel composition of $G$ and $\tilde{G}$, and to verify that all the steps possible in $\tilde{G}$ are also possible in $G$. In Figure~\ref{figextension}, $\tilde{G}$ is extended with the states of $G$, written in the states of $\tilde{G}$. Note that after reading the string $ca$, the generator $\tilde{G}$ is in a state from which $b$ can be read, but $G$ being in state 5 can read only $c$. Because of this symbol $b$, the language $L_m(G)$ is not conditionally decomposable. The reader can verify that adding $b$ to $E_k$ results in the situation where $L_m(G)$ is conditionally decomposable with respect to $E_1$, $E_2$, and $E_k\cup\{b\}$.
  $\hfill\diamond$\end{example}
  
  Let a language $L$ be represented by the minimal generator $G=(Q,E,\delta,q_0,F)$ with the total transition function $\delta$ such that $L_m(G)=L$. Assume that alphabets $E_1$, $E_2$, $E_k$ satisfy $E_1\cap E_2 \subseteq E_k\subseteq E_1\cup E_2=E$, and see Algorithm~\ref{extension}.
  \begin{algorithm}[ht]
  \caption{Extension of $E_k$.}\label{extension}
    \begin{algorithmic}[1]
    \Procedure{Extension}{$G,E_1,E_2,E_k$}
      \State Compute $\tilde G$
      \State Compute ${\tt trim}(\tilde{G})$
        \Comment Now, we compute, step-by-step, the generator $H$ for ${\tt trim}(\tilde{G}) \parallel G$.
      \State Set $Q_H=\{((q_{0,1},q_{0,2}),q_0)\}$, a pair of initial states of $\tilde{G}$ and $G$
        \Comment The initial state of $H$.
      \ForAll {$((q_1,q_2),q)\in Q_H$}
        \ForAll {$a\in E\cup\tilde{E}_{1-k}\cup\tilde{E}_{2-k}$}
          \If {$a\in\tilde{E}_{1-k}\cup\tilde{E}_{2-k}$ and $\delta_{\tilde{G}}((q_1,q_2),a)!$}
            \State $\delta_H(((q_1,q_2),q),a)= (\delta_{\tilde{G}}((q_1,q_2),a),q)$
          \EndIf
          \If {$a\in E$ and $\delta_{\tilde{G}}((q_1,q_2),a)!$}
            \If {$\delta(q,a)!$}
              \State $\delta_H(((q_1,q_2),q),a)= (\delta_{\tilde{G}}((q_1,q_2),a),\delta(q,a))$
            \Else
              \If {$a\notin E_k$} 
                \State $E_k=E_k\cup\{a\}$ 
                \Comment{$a$ is allowed in $\tilde G$, but not in $G$; adding it to $E_k$ solves this problem.}
              \Else
                \State $E_k=E_k\cup\{b\}$, where $b\in E\setminus E_k$
              \EndIf
              \State Restart the procedure with the updated set $E_k$.
            \EndIf
          \EndIf
        \EndFor
      \EndFor
      \State \Return $E_k$.
    \EndProcedure
    \end{algorithmic}
  \end{algorithm}
  To prove that the algorithm is correct, note that it computes $L_m(\tilde{G})\cap \tilde{P}^{-1}(L_m(G))$ because $L_m(\tilde{G})=L_m({\tt trim}(\tilde{G}))$. If the condition on line 11 is always satisfied, it means that  $L_m(\tilde{G})\cap \tilde{P}^{-1}(L_m(G))=L_m(\tilde{G})$. In other words, $L_m(\tilde{G})\subseteq\tilde{P}^{-1}(L_m(G))$, which means by Theorem~\ref{thm3} that $L_m(G)$ is conditionally decomposable. On the other hand, if the condition on line 11 is not satisfied, there exists a string $s\in L({\tt trim}(\tilde{G}))=\overline{L_m(\tilde{G})}$ such that $\tilde{P}(s)\notin L(G)=\overline{L_m(G)}$, where the last equality follows from the assumption that $G$ is minimal. This implies that $\tilde{P}(L_m(\tilde{G}))\not\subseteq L_m(G)$, hence $L_m(G)$ is not conditionally decomposable by Theorem~\ref{lem1}. The algorithms halts because we have only a finite number of events to be added to $E_k$, and the language is conditionally decomposable for $E_k=E_1\cup E_2$.
  
  The complexity of this algorithm is $O(|E|^2\cdot |Q|^3)$, which follows from the complexity of Algorithm~\ref{IsCD} and the fact that, in the worst-case, we have to run the algorithm $|E|$ times. Note that the resulting extension depends on the order the states of $G$ and $\tilde G$ are examined. It should be clear that, in general, there might be different extensions (with respect to set inclusion) that correspond to different orders. This is a typical issue with algorithms extending the event sets in such a way that a particular property becomes true. There are examples where the algorithm does not construct the minimal possible extension. Note that to construct the minimal extension (with respect to set inclusion) is an NP-hard problem~\cite{JDEDS}.

  The following example demonstrates the situation where the event that causes the problem on line~11 already belongs to $E_k$. Thus, to solve the problem, another symbol from $E\setminus E_k$ must be added to $E_k$.
  \begin{example}
    Consider the generator $G$ depicted in Figure~\ref{fig02x}\subref{GGx}, where the corresponding alphabets are $E_1=\{a_1,u\}$, $E_2=\{a_2,u\}$ and $E_k=\{u\}$. The isomorphic generators $f_{1+k}(G)$ with renamed event $a_2$ and $f_{2+k}(G)$ with renamed event $a_1$ are depicted in Figure~\ref{fig02x}\subref{GG1x} and Figure~\ref{fig02x}\subref{GG2x}, respectively.
    \begin{figure}[ht]
      \centering
      \subfloat[Generator $G$.]{\label{GGx}\includegraphics[scale=.9]{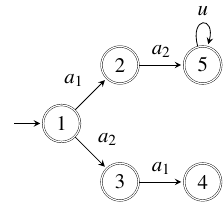}}
      \qquad
      \subfloat[Generator $f_{1+k}(G)$.]{\label{GG1x}\includegraphics[scale=.9]{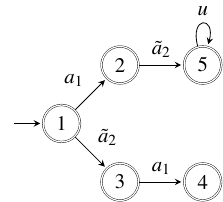}}
      \qquad
      \subfloat[Generator $f_{2+k}(G)$.]{\label{GG2x}\includegraphics[scale=.9]{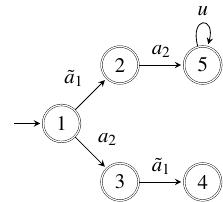}}
      \caption{Generators $G$, $f_{1+k}(G)$ and $f_{2+k}(G)$.}
      \label{fig02x}
    \end{figure}
    The parallel composition $\tilde{G}= f_{1+k}(G)\parallel f_{2+k}(G)$
    \begin{figure}[hb]
      \centering
      \includegraphics[scale=.33]{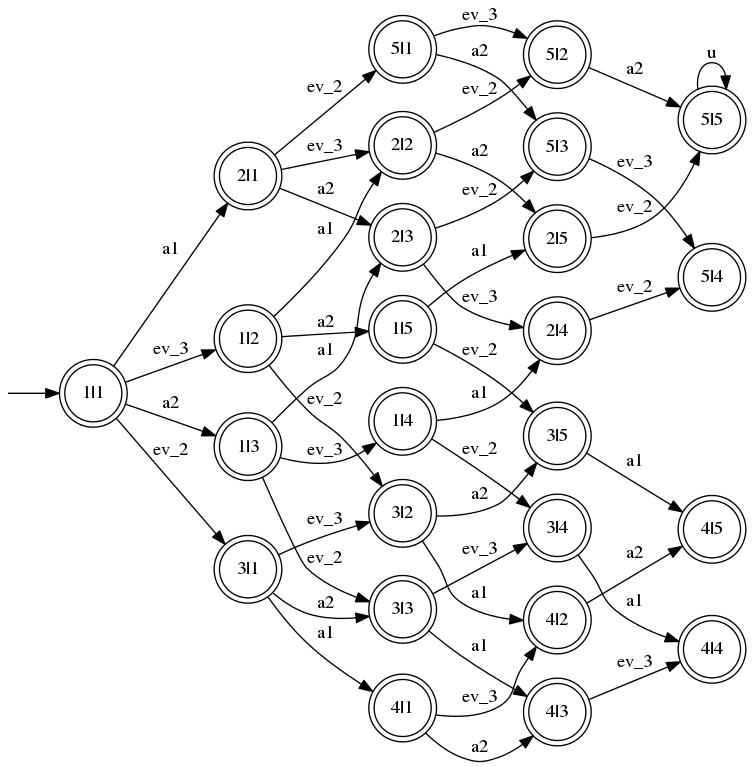}
      \caption{Generator $\tilde G$, where $E_1=\{u,a_1\}$, $E_2=\{u,a_2\}$, and $\tilde{a_1}=ev_3$, $\tilde{a_2}=ev_2$}
      \label{fig01x}
    \end{figure}
    and one can see that the string $\tilde a_1 a_2a_1 \tilde a_2 u$ belongs to $L_m(\tilde{G})$, but does not belong to $\tilde{P}^{-1}(L_m(G))$. By Theorem~\ref{lem1}, $L_m(G)$ is not conditionally decomposable with respect to $E_1$, $E_2$, $E_k$. 
    However, since $u$ belongs to $E_k$, another event that does not belong to $E_k$ must be added to $E_k$. Namely, either $a_1$ or $a_2$. 
  $\hfill\diamond$\end{example}

  This opens a field for potentially interesting heuristics. Indeed, to solve the problem on line~11, it does not make sense to add to $E_k$ an event that does not appear on the path leading to the problematic state. Thus, one could, for instance, store the last visited event from $E\setminus E_k$ that leads the generator $H = \tilde G\parallel G$ to a state where the problem was discovered. This event is then added to $E_k$ on line~17.

\section{Relationship of nonblockingness of coordinated systems to conditional decomposability}\label{cdandnonblocking}
	In this section, we study the relation between conditional decomposability and nonblockingness of coordinated discrete-event systems. A coordinated modular discrete-event system is a system composed (by parallel composition) of two or more subsystems. In this section, we consider the case of one central coordinator. Let $n\ge 2$, and let $G_i$, $i=1,2,\dots,n$, be generators over the respective alphabets $E_i$, $i=1,2,\dots,n$. The coordinated system $G$ is defined as $G=G_1\parallel G_2\parallel \ldots \parallel G_n\parallel G_k$, where $G_k$ is the coordinator over an alphabet $E_k$, which contains all shared events; namely, $E_{s}\subseteq E_k$, where $E_s$ is the set of all events that are shared by two or more components, defined as \[E_s=\bigcup_{i,j\in \{1,\dots ,n\}}^{i\neq j} (E_i\cap E_j)\,.\] This is a standard assumption in hierarchical decentralized control where the coordinator level plays a role of the high (abstracted) level of hierarchical control.
	
	In the following theorem, we show the relation between nonblockingness of a coordinated system and conditional decomposability of that system. First, however, we need the following auxiliary lemmas.
  \begin{lemma}[Proposition 4.1 in \cite{FLT}]\label{proj}
    Let $L\subseteq E^*$ be a language and $P_k: E^*\to E_k^*$ be a projection with $E_k\subseteq E$, for some alphabet $E$. Then, $P_k(\overline{L})=\overline{P_k(L)}$.
  \end{lemma}

  \begin{lemma}\label{simple}
    Let $E$ be an alphabet, $L\subseteq E^*$ be a language, and $P_k: E^*\to E_k^*$ be a projection with $E_k\subseteq E$, for some alphabet $E$. Then, $L\parallel P_k(L)=L$.
  \end{lemma}
  \begin{proof}
    By definition, $L\parallel P_k(L) = L\cap P_k^{-1}P_k(L)$, and it is not hard to see that $L\subseteq P_k^{-1}P_k(L)$.
  \qed\end{proof}

  \begin{theorem}\label{thm1}
    Let $n\ge 2$, and let $G_i$, for $i=1,2,\ldots,n$, be generators over the alphabets $E_i$, $i=1,2,\ldots,n$, respectively. Let $G_k$ be a generator over an alphabet $E_k$ such that $E_{s} \subseteq E_k \subseteq \bigcup_{i=1}^{n} E_i$. Then, the coordinated system $G=G_1 \parallel G_2\parallel \ldots \parallel G_n \parallel G_k$ is nonblocking if and only if the following conditions both hold:
    \begin{enumerate}
      \item $G_i \parallel G_k \parallel \bigparallel _{j\neq i} P_{k}(G_j)$, for all $i=1,2,\ldots,n$, are nonblocking and
      \item $\overline{L_m(G)}$ is conditionally decomposable with respect to the alphabets $E_1, E_2,\ldots, E_n, E_k$.
    \end{enumerate}
  \end{theorem}
	\begin{proof}
		The following always holds for all $i=1,2,\ldots,n$, $n\ge 2$:
    \begin{equation}\label{eqq}
      \begin{split}
        \overline{L_m(G)}
        & \subseteq P_{1+k}(\overline{L_m(G)}) \parallel \ldots \parallel  P_{n+k}(\overline{L_m(G)})\\
        & \subseteq P_{1+k}(L(G)) \parallel \ldots \parallel  P_{n+k}(L(G))\\
        & = L(G_1 \parallel G_k \parallel P_k(G_2 \parallel G_3 \parallel \ldots \parallel G_n)) \\
        & \quad \parallel  L(G_2 \parallel G_k \parallel P_k(G_1 \parallel G_3 \parallel \ldots \parallel G_n))\\
        & \qquad \vdots\\
        & \quad \parallel  L(G_n \parallel G_k \parallel P_k(G_1 \parallel G_2 \parallel \ldots \parallel G_{n-1}))\\
        & = L(G)\,,
      \end{split}
    \end{equation}
    where the last equation follows from the idempotent property of the parallel composition and Lemma~\ref{simple}. If the language $\overline{L_m(G)}$ is nonblocking, then the inclusions become equalities. Thus, from the first equality, we get that the language $\overline{L_m(G)}$ is conditionally decomposable as required in item 2 of the theorem. Similarly, for all $i=1,2,\ldots,n$,
    \begin{align*}
      P_{i+k}(\overline{L_m(G)}) & = \overline{P_{i+k}(L_m(G))}
        = \overline{L_m(G_i\parallel G_k) \parallel P_{i+k}(L_m(\parallel _{j\neq i} G_j))}\\
      & = \overline{L_m(G_i \parallel G_k \parallel P_{i+k}(\parallel _{j\neq i} G_j)}\\
      & \subseteq \overline{L_m(G_i)} \parallel \overline{L_m(G_k)} \parallel \overline{P_{i+k}(\parallel _{j\neq i} L_m(G_j))}\\
      & \subseteq L(G_i) \parallel L(G_k) \parallel P_{i+k}(\parallel _{j\neq i} L(G_j))\\
      & = P_{i+k}(L(G))\,,
    \end{align*}
    where the first equality holds by Lemma~\ref{proj}, the second equality holds by Lemma~\ref{lemma:Wonham} because we project to the alphabet $E_i\cup E_k$ that includes the intersection of $E_i\cup E_k$ and $\bigcup_{j\neq i} E_j$, namely $E_k$. Finally, the last equality holds by the same argument as the second equality. Hence, if the global plant is nonblocking, the inclusions become equalities, which means that the subsystems $G_i \parallel G_k \parallel P_{i+k}(\parallel _{j\neq i} G_j)=G_i \parallel G_k \parallel \bigparallel _{j\neq i} P_{k}(G_j)$ are nonblocking. 
    
		On the other hand, from the assumptions 1 and 2 we immediately get that both inclusions in $(\ref{eqq})$ are equalities. Thus, the implication holds.
	\qed\end{proof}
  
	Note that Condition 2 of Theorem~\ref{thm1} does not hold in general because one inclusion of conditional decomposability, namely $\overline{L_m(G)} \subseteq P_{1+k}(\overline{L_m(G)})\parallel P_{2+k}(\overline{L_m(G)})$, can be strict. Thus, the prefix closure of the marked language $\overline{L_m(G_1\parallel G_2\parallel G_k)}$ of the coordinated system consisting of subsystems $G_1$ and $G_2$ and a coordinator $G_k$ is not in general conditionally decomposable with respect to alphabets $E_1$, $E_2$, $E_k$ as demonstrated in the following example.
	\begin{example}\label{ex1}
		Consider two subsystems $G_1$ and $G_2$, and a coordinator $G_k$ as depicted in Figure~\ref{fig00}, where the corresponding alphabets are $E_1=\{a,b,d\}$, $E_2=\{a,c,d\}$, and $E_k=\{a,d\}$.
		\begin{figure}[ht]
      \centering
      \subfloat[Generator $G_1$.]{\label{G1}\includegraphics[scale=1]{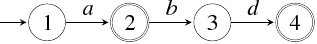}}
      \qquad\qquad
      \subfloat[Generator $G_k$.]{\label{Gk}\includegraphics[scale=1]{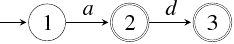}}
      \qquad\qquad
      \subfloat[Generator $G_2$.]{\label{G2}\includegraphics[scale=1]{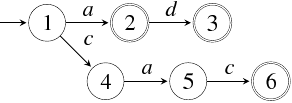}}
      \qquad\qquad\vspace{1mm}
      \subfloat[Generator $G_1\parallel G_2\parallel G_k$.]{\label{G}\includegraphics[scale=1]{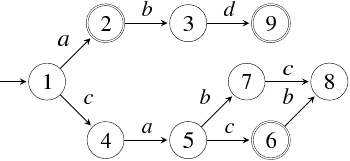}}
      \caption{Generators $G_1$, $G_2$, $G_k$, and $G_1\parallel G_2\parallel G_k$.}
      \label{fig00}
    \end{figure}
		 Then, we can consider the string $cacb$ and see that its projection $P_{1+k}(cacb)=ab$ belongs to the language $P_{1+k}(\overline{L_m(G_1\parallel G_2\parallel G_k)})$, and the projection $P_{2+k}(cacb)=cac$ belongs to the language $P_{2+k}(\overline{L_m(G_1\parallel G_2\parallel G_k)})$. However, this means that the string $cacb$ belongs to the composition $P_{1+k}(\overline{L_m(G_1\parallel G_2\parallel G_k)}) \parallel P_{1+k}(\overline{L_m(G_1\parallel G_2\parallel G_k)})$. On the other hand, the string $cacb$ is not a prefix of any string belonging to the marked language $\overline{L_m(G_1\parallel G_2\parallel G_k)}$ of the coordinated system as is easily seen in Figure~\ref{fig00}\subref{G}. Thus, the language is not conditionally decomposable with respect to alphabets $E_1$, $E_2$, $E_k$.
	$\hfill\diamond$\end{example}

  Note that it follows from~(\ref{eqq}) that conditional decomposability is a weaker condition than nonblockingness. This is because conditional decomposability requires only the first inclusion to be equality, while nonblockingness requires both the inclusions to be equalities.
  The fundamental question is whether it is possible to decide in a distributed way without computing the whole plant whether $\overline{L_m(\parallel_{i=1}^{n} G_i\parallel G_k)}$ is conditionally decomposable. The algorithm described in the previous section requires the computation of the whole plant.

  A specific choice of $L_m(G_k)\subseteq \bigcap_{i=1}^{n}P_k(L_m(G_i))$, respectively $L_m(G_k)=\bigcap_{i=1}^{n}P_k(L_m(G_i))$, yields Corollaries~\ref{cor1} and \ref{cor2} below, respectively.
  \begin{corollary}\label{cor1}
    Let $G_1, G_2,\ldots,G_n, G_k$ be nonblocking generators over the alphabets $E_1,E_2,\dots,E_n, E_k$, respectively, such that $E_{s}\subseteq E_k\subseteq \bigcup_{i=1}^{n} E_i$. Assume that $L_m(G_k)\subseteq \bigcap_{i=1}^{n}P_k(L_m(G_i))$. Then, the coordinated system $G=G_1 \parallel G_2\parallel \ldots \parallel G_n \parallel G_k$ is nonblocking if and only if the following conditions both hold:
    \begin{enumerate}
      \item $G_i\parallel G_k$ are nonblocking, for all $i=1,2,\dots,n$, and
      \item $\overline{L_m(G)}$ is conditionally decomposable with respect to the alphabets $E_1,E_2,\ldots, E_n, E_k$.
    \end{enumerate}
  \end{corollary}
  \begin{proof}
    By the assumption, $L_m(G_k)\subseteq \bigcap_i P_k(L_m(G_i))$. Applying the prefix closure to the previous inclusion results in the inclusion 
    $L(G_k) = \overline{L_m(G_k)} \subseteq \overline{P_k(\bigcap_i L_m(G_i))} \subseteq \bigcap_i P_k(\overline{L_m(G_i)}) = \bigcap_i P_k(L(G_i)) = \bigparallel_i P_k(L(G_i))$. From this, it follows that $L(G_k)\parallel P_k(L(G_i))=L(G_k)$, for $i=1,2,\ldots,n$, which implies that $G_i\parallel G_k\parallel \bigparallel_{j\neq i} P_k(G_j)=G_i\parallel G_k$. Thus, item 1 of Theorem~\ref{thm1} reduces to item 1 of this corollary.
  \qed\end{proof}

  \begin{corollary}\label{cor2}
    Let $G_1, G_2,\ldots,  G_n, G_k$ be nonblocking generators over the alphabets $E_1, E_2,\ldots, E_n, E_k$, respectively, such that $E_s\subseteq E_k\subseteq \bigcup_{i=1}^{n}E_i$, and assume that $L_m(G_k)=\bigcap_{i=1}^{n}P_k(L_m(G_i))$. Then, the coordinated system $G =\ \parallel _{i=1}^n G_i \parallel G_k$ is nonblocking if and only if the following conditions both hold:
    \begin{enumerate}
       \item $G_i\parallel G_k$ are nonblocking, for all $i=1,2,\dots,n$, and
      \item $\overline{L_m(G_1\parallel G_2\parallel \ldots\parallel G_n)}$ is conditionally decomposable with respect to alphabets $E_1,E_2,\ldots, E_n, E_k$.
    \end{enumerate}
  \end{corollary}
  \begin{proof}
    The proof follows immediately from the previous corollary and the fact that $\bigparallel_i L_m(G_i)\parallel L_m(G_k) = \left(\bigparallel_i L_m(G_i)\right) \parallel \left(\bigparallel_i P_k(L_m(G_i))\right)$, which is equal to $\bigparallel_i L_m(G_i)$ by Lemma~\ref{simple}, which reduces item 2 of Corollary~\ref{cor1} to the form of item 2 of this corollary.
  \qed\end{proof}

  The last corollary is particularly interesting because the coordinated modular discrete-event system coincides with the original plant and, therefore, nonblockingness of the original plant itself can be checked using the approach based on a coordinator, provided that we can verify item 2 in a distributed way. 
  
  The approach discussed above is based on projections, and the only known sufficient condition ensuring that the projected automaton is smaller with respect to the number of states than the original one is the observer property mentioned below. This topic requires further investigation because the observer property is only a sufficient condition, not necessary; there are examples of projected automata that are smaller than original automata without the projections satisfying the observer property. For completeness, however, we now discuss the case of projections satisfying the observer property and show that it corresponds to the known results discussed in \cite{FLT} and in references therein.

  Finally, we mention that in practice one central coordinator is particularly useful for loosely coupled subsystems, where the interaction between the subsystems (via synchronisation) is not too strong. Otherwise, a general multilevel hierarchy approach should be adopted, where the subsystems are aggregated into groups that are only loosely coupled. This is, however, very technical and left for a future study.

\subsection{Observer property}
  The previous results are of interest in the case the projected systems $P_k(G_i)$, for $i=1,2,\ldots,n$, are significantly smaller than the original systems $G_i$. So far, the only known condition ensuring this is a so-called {\em observer property}. 

  \begin{definition}[Observer property]\label{observer}
    Let $E_k\subseteq E$ be alphabets. A projection $P_k : E^* \to E_k^*$ is an {\em $L$-observer} for a language $L\subseteq E^*$ if the following holds: for all strings $t\in P(L)$ and $s\in \overline{L}$, if $P(s)$ is a prefix of $t$, then there exists $u\in E^*$ such that $su\in L$ and $P(su)=t$.
  \end{definition}

  The following lemma proves that if the projections are observers, then item 2 of the previous results can be eliminated because it is always satisfied.
  \begin{lemma}\label{lem1c}
    Let $G_i$, $1,2,\ldots,n$, $n\ge 2$, and $G_k$ be generators over the alphabets $E_i$, $i=1,2,\ldots,n$, and $E_k$, respectively, such that $E_s \subseteq E_k\subseteq \bigcup_i E_i$, and denote $G=\bigparallel_i G_i \parallel G_k$. If the projections $P^{i+k}_k$ are $P_{i+k}(L_m(G))$-observers, for $i=1,2,\ldots,n$, then the language $\overline{L_m(G)}$ is conditionally decomposable with respect to $E_i$, $i=1,2,\ldots,n$, and $E_k$.
  \end{lemma}
  \begin{proof}
    By Lemma~\ref{proj}, showing the first equality, it holds in general that
    \begin{equation}\label{eq0b}
      \begin{split}
        \bigparallel_{i=1}^{n} P_{i+k}(\overline{L_m(G)})
        & = \bigparallel_{i=1}^{n} \overline{P_{i+k}(L_m(G))}
          \supseteq \overline{\bigparallel_{i=1}^{n} P_{i+k}(L_m(G))}\\
        & = \overline{\bigparallel_{i=1}^{n} L_m\left(G_i \parallel G_k \parallel \bigparallel_{j\neq i} P_k(G_j)\right)}
          = \overline{L_m(G)}\,.
      \end{split}
    \end{equation}
    The last equality follows from the commutativity of the synchronous product and Lemma~\ref{simple}. By \cite{pcl06}, it holds that $\bigparallel_{i=1}^{n} \overline{P_{i+k}(L_m(G))} = \overline{\bigparallel_{i=1}^{n} P_{i+k}(L_m(G))}$ if and only if $\bigparallel_{i=1}^{n} \overline{P_{k}(L_m(G))} = \overline{\bigparallel_{i=1}^{n} P_{k}(L_m(G))}$, and the later equality is obviously satisfied. Thus, the former equality implies by (\ref{eq0b}) that the language $\overline{L_m(G)}$ is conditionally decomposable with respect to alphabets $E_1$, $E_2$, $E_k$, which was to be shown.
  \qed\end{proof}
  
  As mentioned in the previous proof, when we consider all the assumptions, Feng \cite{FLT} (see also the references therein) has shown that if the projection $P_k$ is an observer for $L_1$ and $L_2$, then $L_1\parallel L_2$ is nonconflicting if and only if $P_k(L_1) \parallel P_k(L_2)$ is nonconflicting. This is generalized to arbitrary components in~\cite{pcl06}. Note that using this property on item 1 of Corollary~\ref{cor2}, together with the previous lemma and the fact that the observers preserve parallel composition, \cite{pcl06}, results in the following corollary, which generalizes the results shown in~\cite{FLT} for two components.
  \begin{corollary}
    Let $G_i$, $1,2,\ldots,n$, $n\ge 2$, and $G_k$ be nonblocking generators over the alphabets $E_i$, $i=1,2,\ldots,n$, and $E_k$, respectively, such that $E_s \subseteq E_k\subseteq \bigcup_i E_i$, and assume that $L_m(G_k)=\bigcap_i P_k(L_m(G_i))$ and $L(G_k)=\bigcap_i P_k(L(G_i))$. Assume that the projections $P^i_k$ are $L_m(G_i)$-observers, for $i=1,2,\ldots,n$. Then, the coordinated system $\bigparallel_i G_i \parallel G_k$ is nonblocking if and only if $G_k$ is nonblocking.
  \end{corollary}

  This works because the projection is an observer. However, there are languages which are conditionally decomposable, but the projections from Lemma~\ref{lem1c} are not observers. For instance, consider a language $L=\{ba,cdb,dcb\}$. It can be verified that $L$ is conditionally decomposable with respect to the alphabets $E_1=\{a,b,c\}$, $E_2=\{a,b,d\}$, and $E_k=\{a,b\}$, and that the projections $P^{i+k}_k$ are not $P_{i+k}(L)$-observers, for $i=1,2$. Note that $P_{1+k}(L)=\{ba,cb\}$ and $P_{2+k}(L)=\{ba,db\}$. Then, for $t=b$ and $s=cb$ (for $i=1$, or $s=db$ for $i=2$), there is no extension of $cb$ such that $P^{1+k}_k(cb)=ba$. Hence, the projections are not observers. For that reason, we consider in this paper a more general assumption that the projections are such that the projected generators are smaller than the original generators. Note that the conditions under which this is true still need to be investigated. Finally, note that for the verification whether the subsystems $G_i\parallel G_k\parallel \bigparallel_{j\neq i} P_k(G_j)$ are nonblocking, the methods presented in~\cite{FM06,FM09} can be used, combined with further usage of Binary Decision Diagrams~\cite{Bryant:1992} or state-tree structures~\cite{MaWonham:2005} to perform the calculations.

\section{Conclusion}\label{conclusion}
  The main contributions of this paper are polynomial-time algorithms for the verification whether a language is conditionally decomposable and for an extension of the coordinator alphabet $E_k$. Our approach to extend the alphabet $E_k$ is based on the successive addition of events to the alphabet $E_k$. Another approach has recently been discussed in~\cite{tmlr}, where the problematic transitions are identified, and the events of these transitions are renamed. From the viewpoint of applications, however, our approach can directly be used in coordination control for which it has primarily been developed. On the other hand, the approach from~\cite{tmlr} has so far no direct applications in the coordination control framework, which is under investigation. Nevertheless, the algorithms presented here can also be used for the approach presented in~\cite{tmlr}.
  
  Particularly valuable is the property that algorithms for checking conditional decomposability of a language with respect to alphabets is linear in the number of alphabets (that corresponds to local controllers in coordination control). No such results are known for co-observability (the notion playing a central role in decentralized control) and the related property of decomposability. It is well-known that co-observability is equivalent to decomposability under some reasonable assumptions on locally controllable and locally observable alphabets. Since conditional decomposability can be seen as decomposability with respect to particular alphabets (enriched by the coordinator events), it appears that our results about conditional decomposability will have impact on decentralized control with communicating supervisors. Indeed, co-observability is ensured by a special types of communication (which corresponds to enriching the sets of locally observable events such that a specification language becomes co-observable) in a similar way as decomposability is imposed by enriching the alphabets of local supervisors.
  
  The paper also compares the property of conditional decomposability to nonblockingness of a coordinated system. The current low complexity tests of practical interest are based on the observer property because it is the only known condition ensuring that the projected generator is smaller than the original one. However, this is only a sufficient condition and further investigation is needed. It is our plan to further investigate the construction procedures for designing coordinators for nonblockingness that are as small as possible and we will combine these results with those obtained in coordination control for safety so that both nonblockingness and safety issues can be efficiently handled using coordination control.

\section*{Acknowledgments}
  The authors gratefully acknowledge very useful suggestions and comments of the anonymous referees.
  The research has been supported by the GA{\v C}R grants P103/11/0517 and P202/11/P028, and by RVO: 67985840.
  
\bibliographystyle{model1-num-names}
\bibliography{biblio}

\end{document}